\documentclass[11pt]{article}

\usepackage[journal=jquant,stage=submission,logging=stdout]{jsty3-author}

\AtBeginDocument{\nolinenumbers}

\usepackage{amsthm}
\usepackage{braket}

\newtheorem{proposition}{Proposition}
\newtheorem{remark}{Remark}
\newtheorem{example}{Example}

\title{Bounds on Entanglement Dynamics from Krylov-Space Spreading}

\author{a,b}{Swati Choudhary}{swatichoudhary@hri.res.in}{}
\author{b,c}{Sukrut Mondkar}{sukrutmondkar@hri.res.in}{}
\author{b,c}{Ujjwal Sen}{ujjwal@hri.res.in}{}

\affiliation{a}{
Center for Quantum Science and Technology (CQST) and
Center for Computational Natural Sciences and Bioinformatics (CCNSB),
International Institute of Information Technology Hyderabad,
Prof. CR Rao Road, Gachibowli, Hyderabad 500032, Telangana, India
}

\affiliation{b}{
Harish-Chandra Research Institute,
Chhatnag Road, Jhunsi, Prayagraj 211019, India
}

\affiliation{c}{
Homi Bhabha National Institute,
Training School Complex, Anushakti Nagar,
Mumbai 400094, India
}

\abstract{
A recently developed notion, known as spread complexity, captures how a quantum state spreads in Krylov space under unitary evolution. Related Krylov-space approaches are increasingly being investigated for applications in quantum control, simulation, and metrology, as well as are being considered for implementation in near-term quantum platforms. However, its connection to fundamental quantum resources such as entanglement and quantum coherence remains unclear.
We show that the dynamics of entanglement is constrained by the extent of spreading in Krylov space. For multipartite systems, we relate state delocalization in the Krylov basis, quantified by the inverse participation ratio, to geometric measures of multipartite quantum correlations. Furthermore, we derive analytical relations between the quantum coherence of the initial state in the energy eigenbasis and spread complexity for qubit and qutrit systems. Our results make progress towards understanding the subtle interplay between the dynamical spreading in Krylov space and fundamental quantum resources.
}

\begin{document}

\section{Introduction}\label{Intro}
Understanding how quantum resources evolve under dynamical processes is an important problem in quantum information science. Quantities such as entanglement~\cite{HHHH2009,GT2009,Das2019} and quantum coherence~\cite{A2006,BCP2014,WY2016,SAP2017} play a fundamental role in enabling quantum advantages in communication, computation, and metrology. A natural question therefore arises: how are these resources constrained and redistributed during unitary quantum evolution? Addressing this question requires tools capable of characterizing not only the presence of quantum correlations, but also the manner in which quantum states spread through 
the state space over time.

Recently, the notion of quantum complexity \cite{W2009,Osborne_2012,FCM2017,BBCCHHH2026} has emerged as a central concept in several areas of physics, including  study of quantum chaos and integrability in many-body quantum systems~\cite{Roberts:2016hpo, Zhuang:2019jyq, Balasubramanian:2019wgd, ABHKMM2020, Yang:2019iav, Haferkamp:2021uxo, BSN2022,HMTW2023,EJX2023, Craps:2023rur,NMBB2024,SJP2024,SRW2024,BMW2025,NS2023,M2024,RGSS2022,RGSS2022CHAOS}. Among the various notions of complexity, a particularly useful one is provided by spread complexity~\cite{BCMW2022} (and its operator analogue, Krylov complexity~\cite{PCASA2019,RGSS2021,CJLPQ2024}), which characterizes the spreading of a quantum state in the Krylov basis generated by the Hamiltonian and the initial state through the Lanczos algorithm~\cite{L1950}. Within this approach, unitary evolution can be mapped onto an effective hopping dynamics along a semi-infinite chain, where the corresponding amplitudes encode the delocalization of the evolving state in Krylov space~\cite{PCASA2019}. Unlike circuit complexity~\cite{N2005,NDGD2006,W2009,BCHKP2021}, which depends on a choice of elementary gates,
 spread complexity provides an intrinsic characterization of Hamiltonian-driven quantum evolution and has been extensively studied in the context of  many-body dynamics, integrability, and quantum chaos~\cite{BDKP2020,BDKLP2021,BMS2021,AB2022,CMP2022,BGN2023,CEP2024,G2024,CHJJKN2024,NMMDC2025,SSN2025,SV2025}. Also,~\cite{BNNS2022,BNNS2023,M2023} extends Krylov formalism to Lindblad/open systems.

{Krylov-space complexity measures are also beginning to find relevance in quantum technologies beyond the theoretical study of quantum dynamics. Recent works have used them to design efficient control protocols for analog quantum simulators~\cite{ZTZL2026}, evaluate quantum-reservoir performance~\cite{DBSRCW2024}, understand charging advantages in quantum batteries~\cite{KMOR2026}, and identify Trotterization-induced changes in digital quantum simulations~\cite{SMC2025}.}
{Spread complexity has further been shown to be related to the minimal cost of synthesizing target states using unitary gates and beam-splitting operations in~\cite{BGMNZ2026}, and experimentally feasible schemes for measuring Krylov complexity on near-term quantum platforms, including superconducting qubits, Rydberg arrays, and trapped ions have been proposed~\cite{XZQL2025}.} {
Krylov-subspace techniques have further been combined with shadow tomography to obtain systematically improvable estimates of the quantum Fisher information~\cite{ZT2025,WZ2026,LW2026} (see also~\cite{CNB2025,ATV2026}).
}
{
Complementing these theoretical developments, Krylov-subspace diagonalization methods have already been implemented on superconducting quantum processors~\cite{Y2025, Kirby2026}.
Moreover, relations between Lanczos coefficients and experimentally accessible quantum-work statistics provide an indirect route to probing quantities governing Krylov-space spreading~\cite{PPGS2024,GPPS2024,BSM2024,An2015,CM2017}.
Taken together, these developments suggest that Krylov-space methods can provide useful tools for quantum control, simulation, machine learning, metrology, sensing, and energy storage.}

Despite the growing importance of Krylov-space diagnostics, their quantitative relation to standard quantum-information-theoretic resources has received comparatively little attention. In particular, it is not clear to what extent dynamical spreading in Krylov space is related to the  dynamics of entanglement~\cite{HHHH2009,GT2009,Das2019} or to the coherence~\cite{A2006,BCP2014,WY2016,SAP2017} properties of the initial state. Since spread complexity depends explicitly on both the Hamiltonian and the chosen initial state, its connection to entanglement and coherence is nontrivial and cannot, in general, be inferred from simple heuristic arguments. This subtle interplay between Krylov spreading and quantum correlations can be understood as follows. An entangled energy eigenstate evolving under its own Hamiltonian acquires only a global phase and therefore exhibits vanishing spread complexity despite possessing finite entanglement. Conversely, a product state evolving under a non-interacting Hamiltonian may display nontrivial spreading in Krylov space while remaining separable throughout the evolution. These examples reveal that the relationship between entanglement and Krylov spreading is non-trivial and model-dependent, thereby motivating the need for deriving rigorous analytical bounds between these quantities.

Another natural question concerns the role of quantum coherence in governing Krylov-space dynamics. If the initial state is an energy eigenstate, the evolution becomes trivial up to a global phase, and no Krylov spreading occurs. This immediately suggests that coherence of the initial state in the energy eigenbasis is necessary for non-zero spread complexity. However, the precise quantitative relation between these quantities remains unknown. In this work, we address this gap by deriving analytical relations between the coherence of the initial state and spread complexity for qubit and qutrit systems. 
It is important  to note that Ref.~\cite{BCS2025} considered the study of coherence in the Krylov basis.  However, we relate the coherence of the initial state, in the energy eigenbasis, with spread complexity.

Our objective in this work is to establish general analytical constraints connecting Krylov-space dynamics with fundamental quantum-information resources. In particular, we derive bounds relating the entanglement properties of evolved states to their delocalization in Krylov space. For bipartite systems, we obtain bounds connecting entanglement  entropy of the evolved state to  spread complexity. For multipartite systems, we derive upper and lower bounds on the inverse participation ratio, a measure of delocalization in Krylov space, in terms of a class of geometric measures of multipartite quantum correlations. Furthermore, for qubit and qutrit systems, we establish analytical relations between the coherence of the initial state in the energy eigenbasis and the resulting spread complexity under arbitrary Hamiltonian evolution.

Taken together, our results provide a quantum-information-theoretic perspective on Krylov-space dynamics and demonstrate that the growth of spread complexity is quantitatively constrained by underlying quantum resources and vice versa. More broadly, our results establish analytical connections between dynamical delocalization in Krylov space and fundamental resource-theoretic properties of quantum states.

The rest of the paper is organized as follows. In Sec.~\ref{Pre}, we introduce the necessary definitions and notation related to Krylov-basis construction, spread complexity, inverse participation ratio, and relevant entanglement and quantum coherence measures. Sec.~\ref{ent-scram} and Sec.~\ref{coh-spread} present our main results. In Sec.~\ref{ent-scram}, we establish bounds relating bipartite and multipartite entanglement to spread complexity and the IPR in Krylov space, respectively. In Sec.~\ref{coh-spread}, we derive analytical relations between quantum coherence and spread complexity for low-dimensional systems. We conclude in Sec.~\ref{Conclusion}.

\section{Preliminaries}\label{Pre}
In this section, we briefly review the essential tools and definitions employed throughout this work, including spread complexity, entanglement, and quantum coherence, which will be employed in the subsequent analysis.\\

\textit{Shannon Entropy.~\cite{S1948, CT1991,J1957}} The Shannon entropy of a discrete probability distribution $\{p_i\}$ is defined as
\begin{equation}
  H(\{p_{i}\}) = -\sum_i p_i \log p_i .  \end{equation}
It quantifies the uncertainty or information content associated with the distribution.
Higher entropy corresponds to greater unpredictability, while $H=0$ when the outcome is certain. \\

\textit{Von Neumann Entropy~\cite{L1973,W1978,S1995,BBPS1996}.}
Consider a bipartite pure state $\rho_{AB}$ defined on the composite Hilbert space 
$\mathcal{H}_A \otimes \mathcal{H}_B$. 
The reduced density operators corresponding to subsystems $A$ and $B$ are obtained 
by taking the partial trace over the complementary subsystem, namely $\rho_A = \mathrm{Tr}_B(\rho_{AB})$, $\rho_B = \mathrm{Tr}_A(\rho_{AB})$.
The von Neumann entropy of the reduced state is defined as
\begin{equation}
S(\rho_A) := - \mathrm{Tr}(\rho_A \ln \rho_A)
= - \mathrm{Tr}(\rho_B \ln \rho_B).
\end{equation}
It serves as a quantitative measure of entanglement between subsystems $A$ and $B$ for pure bipartite states. In particular, $S(\rho_A)=0$ for separable states and is strictly positive for entangled states.\\
 

\textit{Class of Geometric Measures~\cite{GTB2005,DMA2021}.}
For a pure n-partite ($n>2$) quantum state $\ket{\psi}$, consider
\begin{equation}\label{gm}
\mathcal{G}_{\psi} = 1 - \max_{\ket{\phi}\in \mathcal{F}} 
|\langle \phi | \psi \rangle|^{2},
\end{equation}
where the maximization is performed over states $\ket{\phi}$ belonging to a chosen set 
$\mathcal{F}$. $\ket{\psi}$ is termed as $m-$ producible if it can be written as  $\ket{\psi}=\otimes_{i=1}^{k} \ket{\eta}_i$, where each $\ket{\eta}_i$ corresponds to at most m particles with $k \le n$. If $\mathcal{F}$ is taken to be the set of $(m-1)$-producible states, 
the quantity defined in Eq.~(\ref{gm}) corresponds to the geometric measure of 
$m$-producibility. For $m=2$, the maximization is performed over fully product pure states, and 
Eq.~(\ref{gm}) reduces to the geometric measure of $n$ inseparability for multipartite 
systems~\cite{WG2003,BDSI2008}. On the other hand, for $m=n$, the measure reduces to the generalized 
geometric measure (GGM) of entanglement~\cite{DS2010}.
\\

\textit{Relative entropy between two states~\cite{W1978,V2002,Nielsen2010}.} The quantum relative entropy, also known as the Umegaki relative entropy, is a fundamental quantity in quantum information theory that measures the distinguishability between two quantum states. For two density operators $\rho$ and $\sigma$ acting on a finite-dimensional Hilbert space $\mathcal{H}$, the Umegaki relative entropy is defined as,
\begin{equation}
S(\rho || \sigma) = \mathrm{Tr}\left[\rho \left(\log \rho - \log \sigma \right)\right],
\end{equation}
provided that the support of $\rho$ is contained in the support of $\sigma$. If this condition is not satisfied, the relative entropy is defined to be infinite.
\\

\textit{$\ell_1$-norm of Coherence~\cite{BCP2014,SAP2017}.}
Let $|\psi\rangle = \sum_i c_{i} |i\rangle \in \mathcal{H}$ be a pure state expressed in a fixed reference basis 
$\{|i\rangle\}_{i}$. The $\ell_1$-norm of coherence is defined as,
\begin{equation}
C_{\ell_1}(|\psi\rangle)
= \sum_{i \neq j} |\rho_{ij}|,
\end{equation}
where $\rho = |\psi\rangle\langle\psi|$ and $\rho_{ij} = c_{i} c_{j}^{*}$ are the matrix elements of $\rho$ in the reference basis. For a two-level system,  $|\psi\rangle = c_0 |0\rangle + c_1 |1\rangle$ with $|c_0|^2 + |c_1|^2 = 1$, the $\ell_1$-norm of coherence becomes
\begin{equation}
C_{\ell_1}(|\psi\rangle)
=|\rho_{01}| + |\rho_{10}|
=|c_0 c_1^{*}| + |c_1 c_0^{*}|
=2 |c_0| |c_1|.
\end{equation}
For a general pure state  $|\psi\rangle = \sum_i c_{i} |i\rangle$, 
the pairwise coherence between the levels $i$ and $j$ is defined as $C_{ij} = |c_{i}|\,|c_{j}|$.
It quantifies the strength of superposition between the basis states 
$|i\rangle$ and $|j\rangle$, and corresponds to the magnitude of the 
off-diagonal density matrix element $|\rho_{ij}| = |c_{i} c_{j}^{*}|$. The total $\ell_1$-coherence is the sum of all pairwise contributions~\cite{ZPW2023}:
\begin{equation}
C_{\ell_1}(|\psi\rangle)
=
\sum_{i \neq j} |c_{i}|\,|c_{j}|
=
2 \sum_{i<j} |c_{i}|\,|c_{j}|.
\end{equation}

Thus, each unordered pair $(i,j)$ contributes $2 |c_{i}|\,|c_{j}|$,
and the elementary building block is $C_{ij} = |c_{i}|\,|c_{j}|$, and these we refer to as the pairwise coherence throughout the paper.\\

\textit{Spread complexity~\cite{PCASA2019, BCMW2022, RGSS2025, NMMDC2025}.} 
Consider a quantum state $|\psi(t)\rangle= e^{-iHt}\ket{\psi_0}$ evolving under a Hamiltonian $H$,
and let $\{|k_n\rangle\}_n$ denote the associated Krylov basis generated via the
Lanczos algorithm (given below) from an initial state $|k_0\rangle = |\psi_0\rangle$.
The state can be expanded as,
\begin{equation}
|\psi(t)\rangle = \sum_{n} \phi_n(t)\, |k_n\rangle,
\qquad
\phi_n(t) = \langle k_n | \psi(t) \rangle.
\end{equation}
For a pure initial state $|k_0\rangle = |\psi_0\rangle$, the Krylov basis states generated through the Lanczos algorithm remain pure, since they are constructed via successive applications of the Hamiltonian followed by orthogonalization as follows,
\begin{subequations}
\begin{align}
|k_0\rangle &= |\psi_0\rangle, \quad  \text{set} \quad b_0 = 0 \quad\text{and} \quad |k_{-1}\rangle = 0 \\
a_n &= \langle k_n | H | k_n \rangle \\
|v_{n+1}\rangle &= H|k_n\rangle - a_n |k_n\rangle - b_n |k_{n-1}\rangle \\
b_{n+1} &= \|\,|v_{n+1}\rangle\,\| \\
|k_{n+1}\rangle &= \frac{|v_{n+1}\rangle}{b_{n+1}}
\end{align}
\end{subequations}
The norm used in the Lanczos recursion is the Hilbert-Schmidt norm induced by the inner product,
$\|v_{n+1}\| = \sqrt{\langle v_{n+1} | v_{n+1} \rangle}.$
Thus, $b_{n+1} = \|v_{n+1}\| = \sqrt{\langle v_{n+1} | v_{n+1} \rangle}$.
The algorithm terminates when $b_{n+1}=0$, and subsequently further action of $H$ does not generate new linearly independent states in the Krylov basis. 
For non-interacting Hamiltonians and suitably chosen product initial states, the Krylov basis states may retain a product structure or can become entangled; however, $|\psi(t)\rangle$ written in the Krylov basis will never be entangled.  In contrast, for generic interacting many-body Hamiltonians, repeated action of $H$ generates correlations, and the Krylov basis states typically become entangled. 

The spread complexity is defined as the average position in Krylov space,
\begin{equation}
K(t) = \sum_{n} n\, |\phi_n(t)|^2.
\end{equation}
It quantifies the dynamical spreading of the state along the Krylov chain,
thereby serving as a measure of operator or state complexity growth under
unitary evolution.\\

\textit{Inverse Participation Ratio (IPR)}~\cite{T1974,PKMS2025,KPPTS2025}. 
Given a normalized state expanded in an orthonormal basis $\{\ket{\phi_n}\}$ as $\ket{\psi(t)} = \sum_n x_n(t) \ket{\phi_n}$,
the IPR is defined as,
\begin{equation}
\mathrm{IPR}(t) = \sum_n |x_n(t)|^4.
\end{equation}
The IPR quantifies how localized the state is in the chosen basis. Thus, smaller values of IPR correspond to stronger delocalization. In the context of spread complexity, the natural basis is the Krylov basis $\{\ket{k_n}\}$ generated by the Hamiltonian and the initial state. Writing $\ket{\psi(t)} = \sum_n c_n(t) \ket{k_n}$, the Krylov-space IPR becomes $\mathrm{IPR} = \sum_n |c_n(t)|^4$.
This quantity measures the localization of the evolved state along the emergent Lanczos chain. 
While the spread complexity captures the average position in the Krylov chain, the IPR instead probes the concentration versus dispersion of amplitudes across the Krylov chain.
In particular, exponential operator growth in chaotic systems is typically accompanied by a rapid decay of $\mathrm{IPR}$, signaling delocalization in Krylov space. 
Hence, the IPR provides complementary information to spread complexity. Consider the dimension of the Krylov space to be $d_K$, then it is easy to verify that
$
\frac{1}{d_K} \le \mathrm{IPR} \le 1,
$
where $d_K$ also denotes the number of sites on the Krylov chain. 
The lower bound is attained for a completely delocalized state, characterized by uniform amplitudes \(c_{n} = \frac{1}{\sqrt{d_K}}\) for all \(n\). In contrast, the upper bound corresponds to a fully localized state, in which \(c_{n} = 1\) for some \(n\), while all other coefficients vanish. Throughout this work, the inverse participation ratio is evaluated exclusively in the Krylov basis.


\section{Bounds relating entanglement and scrambling}\label{ent-scram}
We are now in a position to present our main results. 
We begin by establishing a bound relating the von Neumann entropy of the evolved state to the corresponding spread complexity. 
It is worth emphasizing that the entanglement properties of the Krylov basis vectors 
$\{\ket{k_n}\}$ do not, in general, require the underlying Hamiltonian to be interacting. 
Since the Krylov basis is constructed through repeated action of $H$ followed by 
orthogonalization, entanglement may arise at the level of individual basis vectors 
even when $H$ itself is non-interacting, depending on the structure of the initial state 
and the generated subspace. However, for the evolved state $\ket{\psi(t)} = e^{-iHt}\ket{\psi_0}$ to dynamically 
develop entanglement starting from a separable initial state, the Hamiltonian must 
contain interacting terms. Consequently, in the context of Proposition~\ref{biparty entk}, 
the Hamiltonian $H$ is required to be interacting in order for the entropy bound 
to capture nontrivial entanglement generation if the initial state is chosen to be product.\\

\begin{proposition}\label{biparty entk}
Let $\{\ket{k_n}\}_n$ denote the Krylov basis generated from a bipartite initial state under unitary evolution governed by a Hamiltonian $H$. Consider the time-evolved state expressed in the Krylov basis as $\ket{\psi(t)} = \sum_{n} c_n(t)\,\ket{k_n}$. Then, at any time $t$, the von Neumann entropy of the reduced state $\rho_A(t)$ is upper bounded in terms of the entropies of the subsystems of the Krylov basis states and the spread complexity as,
\begin{equation}
S\big(\rho_A(t)\big)\;\le\;d_{K}\left[
\sum_{n=0}^{d_{K}-1} |c_n(t)|^2\,S\big(\rho_A^{(n)}\big)\;+\;f(K)
\right],
\end{equation}
where $f(K) = (1+K)\log(1+K) - K\log K$,
$K$ denotes the spread complexity, and $d_K$ is the dimension of the Krylov space.
\end{proposition}
    
\begin{proof}
It was shown in~\cite{GR2008} that for a set of normalized orthogonal bipartite states $\{\ket{k_n}\}_n$, and a state $\ket{\psi} = \sum_n c_n(t)\,\ket{k_n}$, the von Neumann entropy of the reduced state $\rho_A(t)$ satisfies the following upper bound:
\begin{equation}
S\big(\rho_A(t)\big)\;\le\;
d_{K}\left[\sum_{n=0}^{d_{K}-1} |c_n(t)|^2\, S\big(\rho_A^{(n)}\big)
\;+\;H\big(\{|c_n(t)|^2\}\big)
\right],
\end{equation}
where $H\big(\{|c_n(t)|^2\}\big)$ denotes the Shannon entropy of the probability distribution $\{|c_n(t)|^2\}$. At any time $t$, this distribution is characterized by a fixed spread complexity, given by the mean position on the Krylov chain. For any such distribution with fixed mean, it holds that $H\big(\{|c_n(t)|^2\}\big)
\;\le\;
(1+K)\log(1+K) - K\log K,$
as proven in Appendix~\ref{shannon}. Substituting this bound into the previous expression, we obtain
\begin{equation}
S\big(\rho_A(t)\big)
\;\le\;
d_{K}\left[
\sum_{n=0}^{d_{K}-1} |c_n(t)|^2\, S\big(\rho_A^{(n)}\big)
\;+\;
f(K)
\right],
\end{equation}
where $f(K) = (1+K)\log(1+K) - K\log K$.
\end{proof}

Proposition~\ref{biparty entk} clarifies how the entanglement of the evolved state is constrained by its Krylov-space structure. The upper bound contains two distinct contributions: the average entanglement of the Krylov basis vectors populated during the evolution and a monotonic function of spread complexity, $f(K)$. Since $f(K)$ is monotonically increasing, the entanglement of the evolved state is necessarily limited whenever the Krylov spreading is small and the populated Krylov vectors carry little entanglement. This result does not, however, imply a monotonic relation between entanglement and spread complexity; in general, either quantity may be nonzero while the other vanishes. A particularly transparent case arises when all populated Krylov vectors are product across the bipartition A:B. In this case, $S(\rho_A^{(n)})=0$, for every populated $n$, and the bound reduces to $S\big(\rho_A(t)\big)
\;\le\;
d_{K} f(K)$. Consequently, within this restricted setting,  $S\big(\rho_A(t)\big)>0$ implies non-zero $K(t)$. The converse does not hold: nonzero Krylov spreading alone does not imply entanglement.

We next proceed onto giving upper and lower bound on the IPR of the evolved state in the Krylov basis at any time \textit{t} in terms of a \textit{class of Geometric measures} defined in Sec(~\ref{Pre}).

 \begin{proposition}\label{multiparty entk}

 Given that $\{\ket{k_n}\}_{n}$ denote the Krylov basis generated from some multipartite initial state under unitary evolution governed by Hamiltonian H, the IPR of the evolved state $\ket{\psi(t)}=\sum_{n} c_{n}(t)\ket{k_{n}}$ written in Krylov basis, at any time t, can be upper and lower bounded in terms of a class of  Geometric measures as follows,
  \begin{equation}
\begin{aligned}\label{eq:IRP_bounds_GGM}
\frac{
1 - \sqrt{\,1 - \left[ \mathcal{G}_{\psi}-X \right]^{2}}
}{2}
&\le \mathrm{IPR}
&\le
\frac{
1 + \sqrt{\,1 - \left[ \mathcal{G}_{\psi}-X \right]^{2}}
}{2},
\end{aligned}
\end{equation} with the condition $\mathcal{G}_{\psi} - X \le 1$,
where $\mathcal{G}_{\psi} \equiv \mathcal{G}(|\psi\rangle)$ and $X \equiv \sum_{n}  |c_n|^{4}  \mathcal{G}_{\ket{k_{n}}} $.
 \end{proposition}
  
 \begin{proof} Consider that the evolved state in Krylov basis $\{\ket{k_{n}}\}_{n}$ at any time \textit{t} is written as $\ket{\psi(t)}=\sum_{n} c_{n}(t)\ket{k_{n}}$. For the sake of completeness, now we shall rewrite the lower bound given Fact 1 of~\cite{DMA2021},
\begin{equation}
|\langle \varphi|\Psi\rangle|^2
\le
\sum_{n} |c_n|^2 |\langle\varphi|k_n\rangle|^2
+
2\sum_{n<m} |c_nc_m|
|\langle\varphi|k_n\rangle|
|\langle\varphi|k_m\rangle|.
\end{equation}
With this, we have,
\begin{equation}
   \begin{aligned}
\mathcal{G}_{\psi}&\ge 1-
\max_{|\varphi\rangle}
\Big[
\sum_n |c_n|^2 |\langle\varphi|k_n\rangle|^2
+
2\sum_{n<m} |c_nc_m|
|\langle\varphi|k_n\rangle|
|\langle\varphi|k_m\rangle|
\Big].
\\[1ex]
&\ge1-
\sum_n |c_n|^2
\max |\langle\varphi|k_n\rangle|^2
-
2\sum_{n<m} |c_nc_m|
\max
\big(|\langle\varphi|k_n\rangle|
|\langle\varphi|k_m\rangle|\big).
\end{aligned}
 \end{equation}

Now we know $\mathcal{G}_{\ket{k_{n}}}
=
1-
\max_{|\varphi\rangle} |\langle\varphi|k_n\rangle|^2 $, thus $\max |\langle\varphi|k_n\rangle|^2
=
1-\mathcal{G}_{\ket{k_{n}}}$ and $\max |\langle\varphi|k_n\rangle|
=
\sqrt{1-\mathcal{G}_{\ket{k_{n}}}}$. Hence, $\max |\langle\varphi|k_n\rangle||\langle\varphi|k_m\rangle|
\le
\sqrt{1-\mathcal{G}_{\ket{k_{n}}}}
\sqrt{1-\mathcal{G}_{\ket{k_{n}}}}.$

Substituting these bounds into the above equation, we get,
\begin{equation}
\begin{aligned}
\mathcal{G}_{\psi}
&\ge
1-
\sum_n |c_n|^2 (1-\mathcal{G}_{\ket{k_{n}}})
-
2\sum_{n<m} |c_nc_m|
\sqrt{1-\mathcal{G}_{\ket{k_{n}}}}
\sqrt{1-\mathcal{G}_{\ket{k_{j}}}}
\\[1ex]
&=\sum_n |c_n|^2 \mathcal{G}_{\ket{k_{n}}}
-
2\sum_{n<m} |c_nc_m|
\sqrt{1-\mathcal{G}_{\ket{k_{n}}}}
\sqrt{1-\mathcal{G}_{\ket{k_{m}}}}.
\end{aligned}
\end{equation}

Using the fact that,
\begin{align}
&\left(\sum_n  |c_n|\sqrt{1-\mathcal{G}_{\ket{k_{n}}}}\right)^2  \nonumber \\ 
&=
\sum_n |c_n|^2(1-\mathcal{G}_{\ket{k_{n}}})  +
2\sum_{n<m}
\sqrt{ |c_n|^2  |c_m|^2}
\sqrt{1-\mathcal{G}_{\ket{k_{n}}}}
\sqrt{1-\mathcal{G}_{\ket{k_{m}}}},
\end{align}

we can write the lower bound obtained in~\cite{DMA2021} as $\mathcal{G}_{(|\psi\rangle)}\ge
1-\left(\sum_n  |c_n|\sqrt{1-\mathcal{G}_{\ket{k_{n}}}}\right)^2.$  \\
Using this lower bound, we derive an upper bound on $G_{|\psi\rangle}$ in App.~\ref{A}. It is given by,
\begin{equation}   
\begin{aligned}
\mathcal{G}_{(|\psi\rangle)}
\le \min_{n} \Bigg\{ \,
& |c_n(t)|^{2} \, \mathcal{G}_{\ket{k_{n}}}  + 2 |c_n(t)| 
\sqrt{1 - |c_n(t)|^{2}}
\sqrt{1 -\mathcal{G}_{\ket{k_{n}}}}
\Bigg\}.
\end{aligned}
\end{equation}
For convenience of notation we define $A_{n} \equiv |c_n(t)|^{2} \, \mathcal{G}_{\ket{k_{n}}} + 
  2 |c_n(t)| 
\sqrt{1 - |c_n(t)|^{2}}
\sqrt{1 - \mathcal{G}_{\ket{k_{n}}}}$. 
Now since $A_{n} \ge 0$ and $\sum_{n}|c_{m}|^2=1$, it is trivial to show that $\min_{n} \{A_{n}\} \le \sum_{n} |c_{m}|^2 A_{n} \le \sum_{n} A_{n}$, which implies,
\begin{equation}\label{up}
\begin{aligned}
\mathcal{G}_{(|\psi\rangle)}
&\le \sum_{n} |c_n|^{2} A_n \\
&= \sum_{n} |c_n|^{4} \mathcal{G}_{\ket{k_{n}}} 
+ 2 \sum_{n} |c_n|^{2} 
\sqrt{|c_n|^{2}(1 - |c_n|^{2})}
\sqrt{1 - \mathcal{G}_{\ket{k_{n}}}}.
\end{aligned}
\end{equation} 
Now, using the fact that $\sqrt{1 - \mathcal{G}_{\ket{k_{n}}}}\le 1$ and Cauchy-Schwartz inequality, the second term of the above inequality becomes,
\begin{equation}\label{up}
\begin{aligned}
&\sum_{n} |c_n|^{2}
\sqrt{|c_n|^{2}(1 - |c_n|^{2})}
\sqrt{1 - \mathcal{G}_{\ket{k_{n}}}}
\nonumber \\
&\le
\sum_{n} |c_n|^{2}
\sqrt{|c_n|^{2}(1 - |c_n|^{2})} \\
&\le
\sqrt{\sum_{n} |c_n|^{4}}
\;
\sqrt{\sum_{n} |c_n|^{2}(1 - |c_n|^{2})}\\
&=
\sqrt{
\left( \sum_{n} |c_n|^{4} \right)
\left( 1 - \sum_{n} |c_n|^{4} \right)
} \\
&=
\sqrt{\mathrm{IPR}}\,
\sqrt{1 - \mathrm{IPR}}.
\end{aligned}
\end{equation}
With this in eq.(~\ref{up}) we finally have the following upper bound,
\begin{equation}\label{GGM-UPPER}
\mathcal{G}_{(|\psi\rangle)}
\le
\sum_{n} |c_n|^{4} \mathcal{G}_{\ket{k_{n}}}
+
2 \sqrt{\mathrm{IPR}}\sqrt{1 - \mathrm{IPR}}.
\end{equation}
For notational simplification, let us define $\sum_{n}  |c_n|^{4}  \mathcal{G}_{\ket{k_{n}}} \equiv X$ and $\mathrm{IPR} \equiv Z$. With this, eq.(~\ref{GGM-UPPER}) takes the following simple looking form,
\begin{equation}
\frac{(\mathcal{G}_\psi - X)^2}{4} \le Z(1 - Z).
\end{equation}
Rearranging the terms, we arrive at a quadratic inequality,
\begin{equation}\label{eq:eqnZ}
- \frac{(\mathcal{G}_\psi - X)^2}{4} \ge Z^2 - Z
\;\Rightarrow\;
Z^2 - Z + \frac{(\mathcal{G}_\psi - X)^2}{4} \le 0.
\end{equation}
whose real solution exists if  $1 - (\mathcal{G}_\psi - X)^2 \ge 0  \implies 
(\mathcal{G}_\psi - X)^2 \le 1.$

Therefore, $Z_{\pm}
=\frac{1 \pm \sqrt{\,1 - (\mathcal{G}_\psi - X)^2}}{2} \implies$
\begin{equation}
\frac{
1 - \sqrt{\,1 - \left( \mathcal{G}_\psi - X \right)^2}
}{2}
\;\le\;
Z
\;\le\;
\frac{
1 + \sqrt{\,1 - \left( \mathcal{G}_\psi - X \right)^2}
}{2}.
\end{equation}
    Recalling that IPR $\equiv $Z, we have the final expression as, 
  \begin{equation}
\frac{
1 - \sqrt{\,1 - \left( \mathcal{G}_\psi - X \right)^2}
}{2}
\;\le\;
\text{IPR}
\;\le\;
\frac{
1 + \sqrt{\,1 - \left( \mathcal{G}_\psi - X \right)^2}
}{2}.
\end{equation}  
 \end{proof}

\begin{remark}
Upper bound on the IPR obtained above is tighter than the trivial bound \( \mathrm{IPR} \leq 1 \) whenever $\left( \mathcal{G}_\psi - X \right)^2 > 0$. Therefore, the evolved state cannot be completely localized on a single Krylov vector.
\end{remark}
Spread complexity specifies the mean position of the evolving state along the Krylov chain, but does not fully characterize how the corresponding probability distribution is dispersed. To quantify fluctuations about this mean, we consider the variance of the Krylov-position operator. As shown in Appendix~\ref{var}, this variance admits both upper and lower bounds in terms of the Krylov-basis IPR. Combining these bounds with Proposition~\ref{multiparty entk} further yields constraints on Krylov-position fluctuations in terms of geometric measures of multipartite quantum correlations.
\section{Relating coherence of initial state and spread complexity}\label{coh-spread}
 It is evident that if the initial state is chosen to be an eigenstate of the Hamiltonian generating the evolution, the corresponding spread complexity is identically zero, since the state acquires only a global phase and does not spread in the Krylov basis. Consequently, any initial state expressed as a non-trivial superposition of distinct energy eigenstates necessarily exhibits non-zero spread complexity, as the dynamics induces non-trivial support over higher Krylov vectors. However, the precise mathematical dependence of spread complexity on coherence of initial state in the energy eigenbasis of the Hamiltonian is not immediate. In particular, while non-zero initial state coherence is necessary for non-trivial Krylov growth, the functional relation between these quantities is subtle and requires explicit analysis. In the following two propositions, we establish such relations explicitly for the cases of qubit and qutrit systems.
 \begin{proposition}
     The coherence of the initial pure state written in the eigenbasis $\{\ket{E_0},\ket{E_1}\}$ of evolution governing a non-degenerate Hamiltonian H, for a qubit system, is related to spread complexity K(t) as follows,
     \begin{equation}
         K(t)=C^2_{l_{1}} \sin^2{\frac{\omega t}{2}},
     \end{equation}
      where $C_{l_{1}} $ denotes $l_{1}$ coherence of the initial state and $\omega$ is the energy difference of the energies i.e $\omega=E_1 - E_0$.
 \end{proposition}
 \begin{proof}
     For a given non-degenerate Hamiltonian $H = \sum_{i=0}^{1} E_i \ket{E_i}\!\bra{E_i}$,
let the initial state be a non-trivial superposition of its energy eigenstates,
$\ket{\psi_0} \equiv \ket{K_0}= c_0 \ket{E_0} + c_1 \ket{E_1},|c_0|^2 + |c_1|^2 = 1,$
where $\{\ket{E_i}\}_{i=0}^{1}$ forms a complete orthonormal eigenbasis of $H$. The time-evolved state under unitary dynamics generated by $H$ is
 $\ket{\psi(t)}
= e^{-iHt}\ket{\psi_0}
= \sum_{i=0}^{1} c_{i} e^{-iE_i t}\ket{E_i}.$ We now construct the Krylov basis via the Lanczos algorithm, taking $\ket{K_0} = \ket{\psi_0}$ as the reference vector. The zeroth Lanczos coefficient is defined as $a_0 = \bra{K_0} H \ket{K_0}.$ Using the spectral decomposition of $H$, we obtain $a_0
= \sum_{i=0}^{1} |c_{i}|^2 E_i
= \langle H \rangle_{\psi_0},$ i.e., the energy expectation value in the initial state.
The Lanczos recursion relation gives $\ket{A_1}
= (H - a_0)\ket{K_0}
- b_0 \ket{K_{-1}}.$ Since by definition $\ket{K_{-1}} = 0$, this reduces to $\ket{A_1}
= (H - a_0)\ket{K_0}.$ Using $H\ket{K_0}
= \sum_{i=0}^{1} c_{i} E_i \ket{E_i},$ we obtain $\ket{A_1}
= \sum_{i=0}^{1} c_{i} (E_i - a_0)\ket{E_i}.$ Let's define $\epsilon_i := E_i - a_0,$
so that the unnormalised first Lanczos vector takes the form
\begin{equation}
\ket{A_1}
= \sum_{i=0}^{1} c_{i} \epsilon_i \ket{E_i}.
\end{equation}
The corresponding normalisation coefficient is $b_1
= \sqrt{\braket{A_1|A_1}}
= \left( \sum_{i=0}^{1} |c_{i}|^2 \epsilon_i^2 \right)^{1/2}.$ For the two-level case, introducing $p_i = |c_{i}|^2$, we have $b_1
= \left( p_0 \epsilon_0^2 + p_1 \epsilon_1^2 \right)^{1/2}.$
Using $a_0 = p_0 E_0 + p_1 E_1$ and $\omega := E_1 - E_0$, one finds
\begin{equation}
\epsilon_0 = - p_1 \omega,
\qquad
\epsilon_1 = p_0 \omega,
\end{equation}
which yields $b_1 = \omega \sqrt{p_0 p_1}.$
The first non-trivial Krylov vector is therefore,
\begin{equation}
\ket{K_1}
= \frac{1}{\omega \sqrt{p_0 p_1}}
\sum_{i=0}^{1} c_{i} \epsilon_i \ket{E_i}.
\end{equation}
Substituting the explicit expressions for $\epsilon_i$, we obtain
\begin{equation}
\begin{aligned}
\ket{K_1}
&= \frac{1}{\omega \sqrt{p_0 p_1}}
\left(
- c_0 p_1 \omega \ket{E_0}
+ c_1 p_0 \omega \ket{E_1}
\right),
\\[1ex]
&=
- \frac{c_0 p_1}{\sqrt{p_0 p_1}} \ket{E_0}
+ \frac{c_1 p_0}{\sqrt{p_0 p_1}} \ket{E_1}.
\end{aligned}
\end{equation}

Writing $c_0 = |c_0| e^{i\theta_0}$ and $c_1 = |c_1| e^{i\theta_1}$,
we can express $\ket{K_1}$ as $\ket{K_1}
=
- |c_1| e^{i(\theta_0 - \theta_1)} \ket{E_0}
+ |c_0| \ket{E_1}.$
Up to an irrelevant global phase, this is equivalent to $\ket{K_1}=- c_1^{*} \ket{E_0}
+ c_0^{*} \ket{E_1}.$ It is easy to check that $a_0+a_1=E_0+E_1$.


For completeness, we recall that the time-evolved state is, 
\begin{equation}
\ket{\psi(t)}
= c_0 e^{-iE_0 t} \ket{E_0}
+ c_1 e^{-iE_1 t} \ket{E_1}.
\end{equation}
\begin{equation}
\begin{aligned}
\implies \langle K_1 | \psi(t) \rangle
&= - c_0 c_1 e^{-iE_0 t}
+ c_0 c_1 e^{-iE_1 t} \\
&= c_0 c_1 \left( e^{-iE_1 t} - e^{-iE_0 t} \right) \\
&= -2i\, c_0 c_1
\, e^{-i\frac{(E_0 + E_1)t}{2}}
\sin\!\left( \frac{\omega t}{2} \right),
\end{aligned}
\end{equation}
With this we have everything to compute spread complexity.
\begin{equation}
\begin{aligned}
K(t)&=\sum^{1}_{n=0}n \left| \langle K_1 | \psi(t) \rangle \right|^2 \\
&= \left| \langle K_1 | \psi(t) \rangle \right|^2 \\
&= 4 |c_0|^2 |c_1|^2 
\sin^2\!\left( \frac{(E_1 - E_0)t}{2} \right) \\
&= C_{\ell_1}^{\,2} 
\sin^2\!\left( \frac{\omega t}{2} \right),
\end{aligned}
\end{equation}

where $C_{\ell_1} = 2 |c_0| |c_1|$ is the  $l_{1}$ coherence of the initial state in the eigenbasis of H.
\end{proof}
It is interesting to see that for above mentioned qubit case, IPR in Krylov basis has a nice relation with spread complexity.
\begin{equation}
    \begin{aligned}
        IPR &=  \left| \langle K_0 | \psi(t) \rangle \right|^4 + \left| \langle K_1 | \psi(t) \rangle \right|^4\\
        &=1-2 \left| \langle K_1 | \psi(t) \rangle \right|^2 +2  \left| \langle K_1 | \psi(t) \rangle \right|^4\\
        &=1-2 C_{\ell_1}^{\,2} 
\sin^2\!\left( \frac{\omega t}{2} \right) + 2C_{\ell_1}^{\,4} 
\sin^4\!\left( \frac{\omega t}{2} \right)\\&=1-2K(t)+2K^{2}(t)
    \end{aligned}
\end{equation}

We now consider a three-level system.

\begin{proposition}\label{qutritcoh}
 Consider a qutrit evolving under a non-degenerate Hamiltonian $H$ with energy eigenbasis $\{\ket{E_0},\ket{E_1},\ket{E_2}\}$ and an initial pure state expressed in this eigenbasis. Then the spread complexity $K(t)$ associated with the unitary evolution generated by $H$ can be written explicitly in terms of the pairwise coherences of the initial state in the energy eigenbasis. In particular, $K(t)$ depends only on the off-diagonal contributions corresponding to the pairs $(E_0,E_1)$, $(E_0,E_2)$, and $(E_1,E_2)$ i.e spread complexity in the three-level case is governed by the pairwise coherence structure of the initial state as follows,
    \begin{equation}
    \begin{aligned}
    K(t)&=
\frac{4}{b_1^{2}}
\left|
\sum_{i<j}
C_{ij}^{2}\,
\omega_{ij}\,
\sin\!\left( \frac{\omega_{ij} t}{2} \right)
e^{-i\frac{(E_i + E_j)t}{2}}
\right|^2
\\[1ex] 
&+ 8N^{2} C_{12} C_{23} C_{31}|\Omega (\omega_{12},\omega_{23})|^2,
     \end{aligned}
\end{equation}
where $\Omega(\omega_{12},\omega_{23})=\omega_{23} e^{-i\frac{(E_1 + E_2)t}{2}}
\sin\!\left( \frac{\omega_{12} t}{2} \right)
+
\omega_{12} e^{-i\frac{(E_2 + E_3)t}{2}}
\sin\!\left( \frac{\omega_{23} t}{2} \right)$
, $C_{ij} $ denotes pairwise coherence (as defined in section~\ref{Pre}) of the initial state and $\omega_{ij}$ is the energy difference of the $i^{th}$ and $j^{th}$ energy level i.e $\omega_{ij}=E_i - E_j$.
\end{proposition}
 
\begin{proof}
We consider a non-degenerate qutrit Hamiltonian $H = \sum_{i=0}^{2} E_i \ket{E_i}\!\bra{E_i},$
with $\{\ket{E_i}\}_{i=0}^{2}$ forming a complete orthonormal eigenbasis.  Let the initial pure state be written in this eigenbasis as $\ket{\psi_0} \equiv \ket{K_0}
= \sum_{i=0}^{2} c_{i} \ket{E_i}$ with $\sum_{i=0}^{2} |c_{i}|^2 = 1$. For the sake of simplicity, we consider the coefficients $c_i$'s to be real. Under the unitary dynamics generated by $H$, the time-evolved state is $\ket{\psi(t)}
= e^{-iHt} \ket{\psi_0}
= \sum_{i=0}^{2} c_{i} e^{-iE_i t} \ket{E_i}.$
We now construct the Krylov basis via the Lanczos algorithm with
$\ket{K_0} = \ket{\psi_0}$ as the reference vector. The zeroth Lanczos coefficient is $a_0
= \bra{K_0} H \ket{K_0}
= \langle H \rangle
= \sum_{i=0}^{2} |c_{i}|^2 E_i,$
namely, the expectation value of the Hamiltonian in the initial state. The first non-trivial Krylov vector is obtained from $\ket{K_1}
= \frac{1}{b_1} (H - a_0)\ket{K_0}$,
where the normalization coefficient $b_1$ is,
\begin{equation}
b_1^2
= \braket{A_1|A_1}
= \bra{K_0}(H - a_0)^2\ket{K_0}
= \bra{K_0}(H - \langle H \rangle)^2\ket{K_0}.
\end{equation}

Thus, $b_1^2
= \langle H^2 \rangle - \langle H \rangle^2
= (\Delta H)^2,$
which is precisely the energy variance in the initial state. Explicitly,
\begin{equation}
(\Delta H)^2
=
\sum_{i=0}^{2} |c_{i}|^2 E_i^2
-
\left( \sum_{i=0}^{2} |c_{i}|^2 E_i \right)
\left( \sum_{j=0}^{2} |c_{j}|^2 E_j \right),
\end{equation}
or equivalently,
\begin{equation}
(\Delta H)^2
=
\sum_{i=0}^{2} |c_{i}|^2 E_i^2
-
\sum_{i,j=0}^{2} |c_{i}|^2 |c_{j}|^2 E_i E_j.
\end{equation}
We now reorganize the double sum appearing in the variance.  First, we decompose the second term as,
\begin{equation}
\sum_{i,j} |c_{i}|^2 |c_{j}|^2 E_i E_j
=
\sum_i |c_{i}|^4 E_i^2
+
2 \sum_{i<j} |c_{i}|^2 |c_{j}|^2 E_i E_j,
\end{equation}
Substituting this into the expression for $(\Delta H)^2$, we get
\begin{equation}
(\Delta H)^2
=
\sum_i |c_{i}|^2 E_i^2
-
\sum_i |c_{i}|^4 E_i^2
-
2 \sum_{i<j} |c_{i}|^2 |c_{j}|^2 E_i E_j.
\end{equation}

From now onwards, we shall write  $p_i := |c_{i}|^2$. With this, we finally have,
\begin{equation}
(\Delta H)^2
=
\sum_i p_i (1 - p_i) E_i^2
-
2 \sum_{i<j} p_i p_j E_i E_j,
\end{equation}
with the normalization condition $\sum_i p_i = 1 \implies 1 - p_i = \sum_{j \neq i} p_j\implies p_i (1 - p_i)
= p_i \sum_{j \neq i} p_j$
so that, $\sum_i p_i (1 - p_i) E_i^2
=
\sum_i \sum_{j \neq i} p_i p_j E_i^2.$ Symmetrizing this expression yields, $\sum_i \sum_{j \neq i} p_i p_j E_i^2
=
\sum_{i<j} p_i p_j (E_i^2 + E_j^2)$. Consequently finally we can write $(\Delta H)^2$ as,
\begin{equation}
\begin{aligned}
   (\Delta H)^2&= \sum_{i<j} p_i p_j (E_i^2 + E_j^2)-2 \sum_{i<j} p_i p_j E_i E_j,
   \\[1ex]
   &=\sum_{i<j} p_i p_j (E_i - E_j)^2.
 \end{aligned}
\end{equation}
Recalling that $b_1^2 = (\Delta H)^2$, and defining the pairwise coherences
\(
C_{ij} := \sqrt{p_i p_j} = |c_{i} c_{j}|
\)
and energy gaps
\(
\omega_{ij} := E_i - E_j,
\)
we finally obtain $b_1^2
=
\sum_{i<j} C_{ij}^{\,2}\, \omega_{ij}^{\,2}.$ Recall that the first non-trivial Krylov vector is given by
\begin{equation}
\ket{K_1}
=
\frac{1}{b_1}(H-a_0)\ket{K_0}
=
\frac{1}{b_1}\sum_i c_{i} (E_i-a_0)\ket{E_i},
\end{equation}
while the time-evolved state reads $\ket{\psi(t)}=\sum_i c_{i} e^{-iE_i t}\ket{E_i}.$
We define the spread complexity (for the three-level case) as $K(t)=|\phi_1(t)|^2 + 2|\phi_2(t)|^2$ with $\phi_i(t)=\braket{K_i|\psi(t)}$.
Using the above expressions, we obtain
\begin{equation}
\phi_1(t)
=
\frac{1}{b_1}
\sum_i |c_{i}|^2 (E_i-a_0)e^{-iE_i t}.
\end{equation}
Here,\begin{equation}
\begin{aligned}
E_i-a_0
&=
E_i-\sum_j |c_{j}|^2 E_j \\
&=
E_i-|c_{i}|^2E_i-\sum_{j\neq i}|c_{j}|^2E_j \\
&=
(1-|c_{i}|^2)E_i-\sum_{j\neq i}|c_{j}|^2E_j \\
&=
\left(\sum_{j\neq i}|c_{j}|^2\right)E_i
-\sum_{j\neq i}|c_{j}|^2E_j \\
&=
\sum_{j\neq i}|c_{j}|^2(E_i-E_j).
\end{aligned}
\end{equation}

Substituting $E_i-a_0$ back to expression of $\phi_1(t)$, we get
\begin{equation}
\phi_1(t)
=
\frac{1}{b_1}
\sum_i |c_{i}|^2
\left(
\sum_{j\neq i}|c_{j}|^2 (E_i-E_j)
\right)
e^{-iE_i t},
\end{equation}
or equivalently, $\phi_1(t)
=
\frac{1}{b_1}
\sum_{i}\sum_{j\neq i}
|c_{i}|^2|c_{j}|^2 (E_i-E_j)
e^{-iE_i t}.$

Each unordered pair $(i,j)$ appears twice in the above double sum. Grouping the contributions for a fixed pair $i<j$, we obtain
\begin{equation}
\begin{aligned}
&|c_{i}|^2 |c_{j}|^2 (E_i - E_j)e^{-iE_i t}
+ |c_{j}|^2 |c_{i}|^2 (E_j - E_i)e^{-iE_j t} \\
&=
|c_{i}|^2 |c_{j}|^2 
\left[
(E_i - E_j)e^{-iE_i t}
+ (E_j - E_i)e^{-iE_j t}
\right].\\
&=|c_{i}|^2 |c_{j}|^2 \, \omega_{ij}
\left( e^{-iE_i t} - e^{-iE_j t} \right) \\
&=
-2i |c_{i}|^2 |c_{j}|^2 \, \omega_{ij}
\, e^{-i\frac{(E_i+E_j)t}{2}}
\sin\!\left( \frac{\omega_{ij} t}{2} \right).
\end{aligned}
\end{equation}

Summing over all unordered pairs, we therefore arrive at,
\begin{equation}
\phi_1(t)
=
-\frac{2i}{b_1}
\sum_{i<j}
|c_{i}|^2 |c_{j}|^2 \,
\omega_{ij}
\, e^{-i\frac{(E_i+E_j)t}{2}}
\sin\!\left( \frac{\omega_{ij} t}{2} \right).
\end{equation}

Defining the pairwise coherence amplitudes
\(
C_{ij} := |c_{i}||c_{j}|,
\)
this expression can be written compactly as,
\begin{equation}
\phi_1(t)
=
-\frac{2i}{b_1}
\sum_{i<j}
C_{ij}^{\,2} \,
\omega_{ij}
\, e^{-i\frac{(E_i+E_j)t}{2}}
\sin\!\left( \frac{\omega_{ij} t}{2} \right).
\end{equation}

\begin{equation}
\implies |\phi_1(t)|^2
=
\frac{4}{b_1^2}
\left|
\sum_{i<j}
C_{ij}^{\,2} \,
\omega_{ij}
\sin\!\left( \frac{\omega_{ij} t}{2} \right)
e^{-i\frac{(E_i+E_j)t}{2}}
\right|^2.
\end{equation}

We will now determine the second Krylov amplitude. Let the third Krylov vector be expanded in the energy eigenbasis as $\ket{K_2} = \sum_{i=0}^{2} d_i \ket{E_i}.$  Consider $d_i$'s to be real. Orthogonality with $\ket{K_0}$ implies $\braket{K_0|K_2}=0
\quad \Rightarrow \sum_i c_i d_i = 0$,
while orthogonality with $\ket{K_1}$ gives $\braket{K_1|K_2}=0
\quad \Rightarrow 
\frac{1}{b_1}\sum_i c_i(E_i-a_0)d_i = 0,$
or equivalently, $\sum_i c_i(E_i-a_0)d_i = 0.$
Introduce vectors
\[
\vec d=(d_1,d_2,d_3), 
\vec v_1=(c_1,c_2,c_3),
\vec v_2=(c_1\Gamma_1,\,c_2\Gamma_2,\,c_3\Gamma_3),
\]
with $\Gamma_i:=E_i-a_0$. Therefore, the orthogonality conditions in terms of above introduced vectors become $\vec v_1\!\cdot\!\vec d=0,
\quad
\vec v_2\!\cdot\!\vec d=0.$
Hence, $\vec d$ must be orthogonal to both $\vec v_1$ and $\vec v_2$.  
In three dimensions this determines $\vec d$ (up to normalization) as $\vec d \propto \vec v_1 \times \vec v_2 .$ Evaluating the components explicitly, $d_1=c_2c_3(E_3-E_2)$, $d_2=
c_3c_1(E_1-E_3)$ , $d_3=c_1c_2(E_2-E_1).$ Thus,
\begin{equation}
\vec d
=
N\!\left(
c_2c_3(E_3-E_2),
\;
c_3c_1(E_1-E_3),
\;
c_1c_2(E_2-E_1)
\right),
\end{equation}
where $N$ is a normalization constant.
Recalling that $\ket{\psi(t)}=\sum_i c_i e^{-iE_i t}\ket{E_i},$
the second Krylov amplitude becomes $\phi_2(t)
=\braket{K_2|\psi(t)}
=\sum_i d_ic_i e^{-iE_i t}.$
Substituting the explicit form of $d_i$ and using $\omega_{ij}:=E_i-E_j$, we obtain,

\begin{equation}
\phi_2(t)
=
- N c_1 c_2 c_3
\!\left[
\omega_{23}e^{-iE_1 t}
+
\omega_{31}e^{-iE_2 t}
+
\omega_{12}e^{-iE_3 t}
\right].
\end{equation}

The term in bracket can be simplified further using $e^{-iE_i t}-e^{-iE_j t}=-2i\,e^{-i\frac{(E_i+E_j)t}{2}}
\sin\!\left(\frac{\omega_{ij}t}{2}\right),$
\begin{equation}
\begin{aligned}
&\omega_{23} e^{-iE_1 t}
+ \omega_{31} e^{-iE_2 t}
+ \omega_{12} e^{-iE_3 t}\\
&=
\omega_{23}\!\left(e^{-iE_1 t}-e^{-iE_2 t}\right)
+\omega_{12}\!\left(e^{-iE_3 t}-e^{-iE_2 t}\right).\\
&=\omega_{23}
e^{-i\frac{(E_1+E_2)t}{2}}
\sin\!\left(\frac{\omega_{12}t}{2}\right)
+
\omega_{12}
e^{-i\frac{(E_2+E_3)t}{2}}
\sin\!\left(\frac{\omega_{23}t}{2}\right)\\
&\equiv \Omega (\omega_{12},\omega_{23})
\end{aligned}
\end{equation}
 where in the last line we have defined the resultant expression as $\Omega (\omega_{12},\omega_{23})$. With this we have,
\begin{equation}
\phi_2(t)=
-2i\,N c_1 c_2 c_3 \Omega (\omega_{12},\omega_{23})
\end{equation}
 $\implies |\phi_2(t)|^2=
4N^2 |c_1|^2|c_2|^2|c_3|^2 |\Omega (\omega_{12},\omega_{23})|^2$. In terms of pairwise coherences $C_{ij}=|c_i||c_j|$, this can be written compactly as,
\begin{equation}
|\phi_2(t)|^2
=
4N^2 C_{12}C_{23}C_{31}
|\Omega (\omega_{12},\omega_{23})|^2.
\end{equation}

Collecting the contributions, the spread complexity
\(
K(t)=|\phi_1(t)|^2+2|\phi_2(t)|^2
\)
takes the form
\begin{equation}
\begin{aligned}
K(t)
&=
\frac{4}{b_1^2}
\left|
\sum_{i<j}
C_{ij}^{2}\,
\omega_{ij}\,
\sin\!\left(\frac{\omega_{ij} t}{2}\right)
e^{-i\frac{(E_i+E_j)t}{2}}
\right|^2
\\[6pt]
&\quad
+
8N^2 C_{12}C_{23}C_{31}
|\Omega (\omega_{12},\omega_{23})|^2.
\end{aligned}
\end{equation}
\end{proof}

   Note that $K(t)$ for a qutrit system is composed of sine functions, hence bounded and oscillatory.

It is instructive to note from Proposition~\ref{qutritcoh} that, for a qutrit system, the resulting expression for the spread complexity $K(t)$ depends explicitly on both pairwise and triple coherences of the initial state in the energy eigenbasis. In particular, the contributions involving $C_{ij}$ capture the pairwise coherence between the levels $(i,j)$, while the terms proportional to $C_{12}C_{23}C_{31}$ arise from the simultaneous presence of coherence among all three energy levels. Furthermore, the expression obtained in Proposition~\ref{qutritcoh} also recovers the expected behavior in the fully degenerate case. Indeed, if the Hamiltonian is completely degenerate, then $\omega_{ij}=E_i-E_j=0$ for all pairs $\{i,j\}$. Consequently, all oscillatory terms appearing in the expression of $K(t)$ vanish, implying that the spread complexity identically reduces to zero. This is consistent with the fact that in the absence of energy gaps, the time evolution becomes trivial up to a global phase, and therefore no Krylov spreading occurs. Therefore, energy non-degeneracy is necessary for Krylov growth.\\
Let us now examine the short-time behaviour of the Proposition~\ref{qutritcoh}.

In the short-time limit $t \to 0$, we use the approximations $\sin\!\left(\frac{\omega_{ij}t}{2}\right) \approx \frac{\omega_{ij}t}{2}$ , $e^{-i\frac{(E_i+E_j)t}{2}} \approx 1$.
Substituting these into the expression for $\phi_1(t)$ obtained in Proposition~\ref{qutritcoh}, we obtain
\begin{equation}
\begin{aligned}
|\phi_1(t)|^2 
&\approx \frac{4}{b_1^2}
\left|
\sum_{i<j}
C_{ij}^{2}\omega_{ij}
\left(\frac{\omega_{ij}t}{2}\right)
\right|^2 \\
&=
\frac{4}{b_1^2}
\left|
\frac{t}{2}
\sum_{i<j} C_{ij}^{2}\omega_{ij}^{2}
\right|^2 .
\end{aligned}
\end{equation}
Using previously obtained expression for $b^{2}_1$ i.e 
\(
b_1^2=\sum_{i<j}C_{ij}^{2}\omega_{ij}^{2},
\)
this simplifies to $|\phi_1(t)|^2 \approx b_1^2 t^2 .$
Similarly, expanding the expression for $\phi_2(t)$ in the short-time limit yields
\begin{equation}
|\phi_2(t)|^2
\approx
32 N^2 C_{12}C_{23}C_{31}
\,\omega_{12}^2 \omega_{23}^2 \, t^2 .
\end{equation}

Thus, both $|\phi_1(t)|^2$ and $|\phi_2(t)|^2$ scale quadratically in time. Since the spread complexity is given by $K(t)=|\phi_1(t)|^2+2|\phi_2(t)|^2,$
it follows that $K(t)\propto t^2$, which is usually the standard case for spread complexity at small times~\cite{G2022}. 

\begin{remark}
If we consider the case where one of the amplitudes of the initial state vanishes, for instance, that $c_3=0$. In this case, the pairwise coherences involving the third level vanish, namely $C_{23}=0$ and $C_{31}=0$. Consequently, the expression for $\phi_2(t)$ vanishes identically, i.e.\ $\phi_2(t)=0$. Therefore, the spread complexity reduces to the contribution coming solely from $\phi_1(t)$, and the resulting expression coincides with that obtained for the qubit case.
\end{remark}

\section{Conclusion}\label{Conclusion}
The growth of spread complexity in quantum dynamics reflects the progressive spreading of an initial state across Hilbert space. 
This spreading arises from the repeated action of the Hamiltonian, which generates higher Krylov modes and leads to an increase in spread complexity. Such dynamical delocalization naturally accompanies the generation of quantum correlations. For bipartite systems, we have shown that the entanglement of the evolved state is upper bounded in terms of the entanglement of the Krylov basis vectors and the spread complexity. In the case of multipartite systems, both lower and upper bounds are obtained for the inverse participation ratio, a quantifier of the delocalization of a quantum state in the Krylov basis, in terms of the geometric measures.  As it is well known that the spread complexity depends on the initial state and the Hamiltonian, we find analytic relations for qubit and qutrit systems 
that relate the initial state’s quantum coherence in the energy eigenbasis and the spread complexity of the time-evolved state for any given arbitrary Hamiltonian. 
\\

\begin{acknowledgments}
SC and US acknowledge partial support from the Department of Science and Technology,
Government of India, through the QuEST grant with Grant
No.~DST/ICPS/QUST/Theme-3/2019/120 via I-HUB QTF of
IISER Pune, India. US also acknowledges financial support from the
Anusandhan National Research Foundation (ANRF), Government of India,
under Grant No.~ANRF/ARG/2025/004617/PS.
\end{acknowledgments}
\appendix

\section{Proof of $H(\{ |c_{n}|^2\}) \le (1+K) \log (1+K)- K \log K$}\label{shannon}
Spread complexity of a pure quantum state $\ket{\psi}=\sum_{n}c_{n}(t)\ket{k_n}$ is $K(t)=\sum_{n}n|c_{n}(t)|^2 $. At a given t, we have this fixed mean(spread complexity) as a constraint along with the normalisation constraint. Shannon entropy at a given fixed t is given by the distribution $\{|c_{n}|^2\}$ is $\sum_{n} p_{n} \log p_{n}$, where $p_{n}=|c_{n}|^2$. The aim is to maximise this Shannon entropy given the two constraints mentioned above.
Lagrangian of the problem becomes,
\begin{equation}
  \mathcal{L}=\sum_{n} p_{n} \log p_{n} + \alpha (\sum_{n}p_n -1)+\beta(\sum_{n}np_{n}-K)  
\end{equation}
Differentiating w.r.t $p_{n}$ and putting $ \frac{\partial  \mathcal{L}}{\partial p_{n}}=0$, 
\begin{equation}
    \frac{\partial  \mathcal{L}}{\partial p_{n}}=-\sum_{n}(\log p_{n}+1)+\alpha \sum_{n} 1 +\beta\sum_{n}n=0
\end{equation}
\begin{equation}
-(\log p_n + 1) + \alpha + \beta n = 0 \qquad \forall n
\end{equation}
\begin{equation}
-\log p_n - 1 + \alpha + \beta n = 0
\end{equation}
\begin{equation}
\log p_n = \alpha + \beta n - 1 = \alpha - 1 + \beta n
\end{equation}
Let us denote $\alpha-1=C' \implies p_{n}=C e^{\beta n}$, where $C=e^{\alpha-1}$. Now $\beta$ has to be less than zero, otherwise $\sum_{n}p_{n}$ will diverge.  Then defining $e^{\beta}=r$ with $0<r<1$, we have $p_{n}=C r^n$ . Normalisation consdition, $\sum_{n}p_{n}=\sum_{n}C r^n=C \frac{1}{1-r}=1 \implies C=1-r \implies p_{n}=(1-r)r^n$ and $K=\sum_{n}n(1-r)r^n=\frac{r}{1-r}$ using geometric sum $\sum_{n}n r^n= \frac{r}{(1-r)^2}$. Writing r in terms of K, we have $p_{n}=\frac{1}{1+K}(\frac{K}{1+K})^n$. With this final form of $p_{n}$, shannon entropy 
\begin{equation}
H\left(\{p_{n}\}\right)
=
-\sum_{n}
\frac{1}{1+K}
\left(\frac{K}{1+K}\right)^n
\log
\left[
\frac{1}{1+K}
\left(\frac{K}{1+K}\right)^n
\right]
\end{equation}
After few steps of simplification, we can easily get $H\left(\{p_{n}\}\right)= (1+K)\log(1+K)-K \log K$, this is basically the maximum Shannon entropy with fixed mean at given t. Therefore, in general $H\left(\{p_{n}\}\right)\le (1+K)\log(1+K)-K \log K$.\\

\section{GM of superoposition of Krylov basis vectors}\label{A}
Consider that a multipartite pure state is written in the Krylov basis,  $|\psi\rangle = \sum_{n} c_n |k_n\rangle=|\psi\rangle = c_m |k_m\rangle + \sum_{n\neq m} c_n |k_n\rangle$
\[
\implies |k_m\rangle = \frac{1}{c_m} |\psi\rangle - \frac{1}{c_m} \sum_{n\neq m} c_n |k_n\rangle
\]
We then define, $|M_m\rangle := \sum_{n\neq m} c_n |k_n\rangle$ and $N_m := \|M_m\| = \sqrt{1-|c_m|^2}$. Let $\alpha_m = \frac{1}{c_m}$, $\beta_m = -\frac{N_m}{c_m}, $ $|\tilde{M}_m\rangle = \frac{|M_m\rangle}{N_m}$. 

therefore, 
\begin{equation}
|k_m\rangle = \alpha_m |\psi\rangle + \beta_m |\tilde{M}_m\rangle
\end{equation}

Applying the lower bound for the superposition of two states gives
\begin{equation}
\mathcal{G}_{|k_m\rangle}
\ge
1 -
\left(
|\alpha_m|\sqrt{1-\mathcal{G}_{\psi}}
+
|\beta_m|\sqrt{1-\mathcal{G}_{\tilde{M}_m}}
\right)^2
  \end{equation}
Using $\sqrt{1-\mathcal{G}_{\tilde{M}_m}} \le 1$, we have 
\begin{equation}
\begin{aligned}
\mathcal{G}_{|k_m\rangle}
&\ge
1 -
\left(
|\alpha_m|\sqrt{1-\mathcal{G}_{\psi}}
+
|\beta_m|
\right)^2
\\[1ex]
&\implies |\alpha_m|\sqrt{1-\mathcal{G}_{\psi}} + |\beta_m|
\ge
\sqrt{1-\mathcal{G}_{|k_m\rangle}}
\\[1ex]
&\implies \sqrt{1-\mathcal{G}_{\psi}}
\ge
\frac{1}{|\alpha_m|}
\left(
\sqrt{1-\mathcal{G}_{|k_m\rangle}))} - |\beta_m|
\right)
\\[1ex]
&\implies \mathcal{G}_{\psi} \le 1 -
\left[
\frac{1}{|\alpha_m|}
\left(
\sqrt{1-\mathcal{G}_{|k_m\rangle}} - |\beta_m|
\right)
\right]^2
\\[1ex]
&\implies \mathcal{G}_{\psi}
\le
1 -
\left[
|c_m|\sqrt{1-\mathcal{G}_{|k_m\rangle}}
-
\sqrt{1-|c_m|^2}
\right]^2
\end{aligned}
\end{equation}

So just like this $|\psi\rangle = c_m |k_m\rangle + \sum_{n\neq m} c_n |k_n\rangle$, we can do this for all $m$ and find for which $m$ we get the minimum. That minimum value would be the upper bound i.e $\mathcal{G}_{\psi}=
\min \{A_1, A_2, \ldots, A_n\}$, where $A_m=1-\left[
|c_m|\sqrt{1-\mathcal{G}_{|k_m\rangle}}
-
\sqrt{1-|c_m|^2}
\right]^2$.\\

\section{Bounds on variance of spread complexity}\label{var}

We consider the Krylov decomposition $|\psi(t)\rangle = \sum_{n=0}^{D-1} c_n(t)\,|K_n\rangle$ with probabilities $p_n(t)=|c_n(t)|^2$ satisfying $\sum_{n} p_n = 1.$ The Krylov operator is defined as,
\begin{equation}
\hat K = \sum_{n=0}^{D-1} n\,|K_n\rangle\langle K_n|.
\end{equation}
Its expectation value and variance are,
\begin{equation}
C(t)=\langle \hat K\rangle = \sum_n n p_n,
\qquad
(\Delta K)^2 = \langle \hat K^2\rangle - \langle \hat K\rangle^2.
\end{equation}

Since $\langle \hat K^2\rangle = \sum_n n^2 p_n,$
the variance can be written as,
\begin{equation}
(\Delta K)^2 = \sum_n n^2 p_n - \left(\sum_n n p_n\right)^2.
\end{equation}

Expanding the second term and symmetrizing,
\begin{equation}
\begin{aligned}
 (\Delta K)^2
&=
\frac12 \sum_{m,n} (m^2+n^2-2mn)\,p_m p_n   \\
&=\frac12 \sum_{m,n} p_m p_n (m-n)^2.
\end{aligned}
\end{equation}

Let us now find lower bound on this variance. Note that since $(m-n)^2 \ge 1$, for $m\neq n$, we obtain
\begin{align}
(\Delta K)^2
&\ge \frac12 \sum_{m\neq n} p_m p_n \\
&= \frac12 \left(1 - \sum_n p_n^2\right)\\
&=\frac12\left(1-\mathrm{IPR}\right).
\end{align}

For the upper bound, since $0 \le n \le D-1$, we have $|m-n|\le D-1$, implying $(m-n)^2 \le (D-1)^2.$

Using the variance identity,
\begin{align}
(\Delta K)^2
&\le \frac12 (D-1)^2 \sum_{m\neq n} p_m p_n \\
&= \frac12 (D-1)^2 \left(1 - \sum_n p_n^2\right)\\
&=\frac{(D-1)^2}{2}\left(1-\mathrm{IPR}\right).
\end{align}

Combining both the bounds, we have 
\begin{equation}
\frac12\left(1-\mathrm{IPR}\right)
\le (\Delta K)^2
\le \frac{(D-1)^2}{2}\left(1-\mathrm{IPR}\right).
\end{equation}
{These obtained upper and lower bounds on variance can be further reexpressed in terms of multipartite correlations using Proposition~\ref{multiparty entk}.}

\vspace{2cm}

\bibliographystyle{jsty3-author}
\bibliography{k}

\end{document}